\documentclass{amsart}
\usepackage{amsfonts}
\usepackage{amsmath,amsthm}
\usepackage{amsfonts,amssymb}
\usepackage{mathrsfs}
\usepackage{enumerate}
\usepackage{verbatim}

\usepackage{multirow}
\usepackage{booktabs}

\usepackage[numbers,sort&compress]{natbib}

\newtheorem{thm}{Theorem}[section]
\newtheorem{cor}[thm]{Corollary}
\newtheorem{lem}[thm]{Lemma}

\numberwithin{equation}{section}

\newcommand{\al}{\alpha}

\def\vz{\varepsilon}

\def\lz{\lambda}

\def\({\Bigl(}
\def \){ \Bigr)}

\newcommand{\uu}{\mathfrak{u}}

\newcommand{\NN}{\mathbb{N}}

\newcommand{\bsgamma}{\boldsymbol{\gamma}}

\newcommand{\bsj}{\boldsymbol{j}}

 \def\NN{{\mathbb N}}

 \def\RR{{\mathbb R}}

\def\bsj{{\bf j}}

\makeatletter
\@namedef{subjclassname@2020}{\textup{2020} Mathematics Subject Classification}
\makeatother

\begin{document}
\def\RR{\mathbb{R}}
\def\Exp{\text{Exp}}
\def\FF{\mathcal{F}_\al}

\title[] {A note about exponential tractability of linear weighted
tensor product problems in the worst-case setting}

\author{Zirong Liu} \address{ School of Mathematical Sciences, Capital Normal
University, Beijing 100048, China.} \email{2240501011@cnu.edu.cn}

\author{Heping Wang} \address{ School of Mathematical Sciences, Capital Normal
University, Beijing 100048, China.}\email{wanghp@cnu.edu.cn}

\author{Kai Wang} \address{ School of Science, Langfang Normal University,
Langfang 065000, China.} \email{cnuwangk@163.com.}

\keywords{Weighted linear tensor product problem; Exponential
$(s,t)$-weak tractability; Exponential uniform weak tractability;
worst-case setting} 

\begin{abstract}

This paper is devoted to discussing  the weighted linear tensor
product problems in the worst  case setting. We consider
algorithms that use finitely many evaluations of arbitrary
continuous linear functionals. We investigate  exponential $(s,
t)$-weak tractability (EXP-$(s, t)$-WT) with $\max(s,t)<1$ and
exponential  uniform weak tractability (EXP-UWT) under the
absolute or normalized error criterion.  We  solve the problem by
filling the remaining gaps left open on EXP-tractability. That is,
we  obtain necessary and sufficient conditions for EXP-$(s, t)$-WT
with $\max(s, t) < 1$ and for EXP-UWT.
\end{abstract}

\maketitle
\input amssym.def

\section{Introduction and main results}

We study $d$-variate linear problem $S = \{S_d\}_{d\in \Bbb
N}$, where $S_d$ is a linear operator from a  Banach function space $H_d$ into another Banach function space $G_d$, $d$ may be large even huge. We consider
algorithms that use finitely many continuous linear functionals on
$H_d$. The information complexity $n(\vz, S_{d})$ is defined to be
the minimal number of linear functionals needed to find an
approximation of $S_{d}$ to within an error threshold $\vz$.
Tractability is the study of how information complexity $n(\vz,
S_{d})$ depends on $\vz$ and $d$. There are two kinds of
tractability based on polynomial convergence and exponential
convergence. The algebraic tractability (ALG-tractability)
describes how  $n(\vz, d)$ behaves as a function of $d$ and
$\vz^{-1}$, while the EXP-tractability does as one of   $d$ and
$(1+\ln\vz^{-1})$. Recently the  study of ALG-tractability and
EXP-tractability   has attracted much interest, and a great number
of interesting results
 have been obtained (see
\cite{ CW1, DKPW14, DLPW11, GW11, IKPW16a, KW19, NW08, NW10, NW12,
PP14, PPW17, S13, SW15, W19, X15} and the references therein).

 This paper is devoted to discussing EXP-tractability of  linear weighted
tensor product problems in the worst case setting. For unweighted
linear  tensor product problems, the results about the
EXP-tractability were obtained in \cite{HKW19, PP14, PPW17}.
However,
 there are many negative results for exponential tractability. Hence  the authors in
 \cite{KPW} introduced weighted
linear  tensor product problems and obtained  necessary and
sufficient conditions for various notions of EXP-tractability on
weights and univariate singular values. Note that in \cite{KPW}
the cases of EXP-$(s,t)$-weak tractability (EXP-$(s,t)$-WT) with $\max(s, t) < 1$ and EXP-uniform weak tractability (EXP-UWT) were
left open.

  In this paper we investigate  EXP-$(s, t)$-WT with $\max(s, t) < 1$  and EXP-UWT for
the above weighted linear tensor product problems. We  solve the
problem by filling the remaining gaps left open on
EXP-tractability. That is, we obtain necessary and sufficient
conditions for EXP-$(s, t)$-WT with $\max(s, t) < 1$ and EXP-UWT.

\subsection{Weighted linear tensor product problems.}

\

Let $H_1$ be a separable infinite-dimensional Hilbert space with
inner product $\langle \cdot,\cdot\rangle_{H_1}$. We consider a
compact linear operator $S_1 : H_1\to G_1$, where $G_1$ is also a
Hilbert space. Then $W_1 := S_1^* S_1: H_1 \to H_1$ is  a compact
self-adjoint non-negative operator, where $T^*$ is the adjoint
operator of a linear operator $T$. Let $\{(\lz_j, e_j)\}_{j\in
\Bbb N}$ denote the eigenpairs of $W_1$, that is,
$$W_1e_j=\lz_je_j,\ e_j\in H_1,\ \ \lz_1\ge \lz_2\ge\dots\ge 0,$$
and $\{e_j\}_{j\in\Bbb N}$ is an orthonormal basis in $H_1$.
Without loss of generality we assume that $\lambda_1=1$ and
$\lambda_2>0$ in the sequel.  Due to compactness of $S_1$ we have
$\lim\limits_{j\to\infty}\lambda_j=0$. Denote
$$H_{1,0}={\rm span}\, \{e_1\},\ \ H_{1,1}=H_{1,0}^\bot.$$
Then $H_1=H_{1,0}\oplus H_{1,1}.$ Hence for each $f\in H_1$, we have the
unique orthogonal  decomposition $f= f_0+ f_1$, where $ f_0\in H_{1,0}, \  f_1\in H_{1,1}$.

 For $\gamma>0$, we define the weighted Hilbert space $H_{1,\gamma}$ to be the Hilbert space $H_1$ with inner product
$$\langle f,g\rangle_{H_{1,\gamma}}:=\langle f_0, g_0 \rangle_{H_1}+\gamma^{-1}\langle f_1, g_1 \rangle_{H_1}.$$
And the weighted operator $S_{1,\gamma}$ is defined to be the
operator $S_1$ from the weighted Hilbert space $H_{1,\gamma}$ to
$G_1$. It follows that the operator
$W_{1,\gamma}=S_{1,\gamma}^*S_{1,\gamma}: H_{1,\gamma}\to
H_{1,\gamma} $ is a compact linear operator and the eigenpairs of
the operator $W_{1,\gamma}$ are $\{(\lz_{\gamma,j},e_j)\}_{j\in\Bbb
N}$, where
\begin{equation}\label{1.1}
\lambda_{\gamma,j}=\left\{
\begin{array}{ll}
1, & \mbox{ if } j=1,\\
\gamma \lambda_j, & \mbox{ if } j \ge 2.
\end{array}\right.
\end{equation}

 Let $\bsgamma=\{\gamma_j\}_{j \in \NN}$ be a sequence of
non-increasing positive reals, which we refer to as weights in the
following.  For simplicity we assume that $\gamma_j\in(0,1]$. For
$d\in \Bbb N$, the tensor product spaces $H_{d,\bsgamma}$ and $G_d$ are denoted by
$$H_{d,\bsgamma}=H_{1,\gamma_1}\otimes
H_{1,\gamma_2}\otimes\dots\otimes H_{1,\gamma_d},\ \
G_d=G_1\otimes \dots\otimes G_1,$$
respectively, the tensor product operator $S_{d,\bsgamma}: H_{d,\bsgamma}\to G_d$ is represented by
$$S_{d,\bsgamma}=S_{1,\gamma_1}\otimes\dots\otimes
S_{1,\gamma_d}.$$ Clearly,
$W_{d,\bsgamma}:=S_{d,\bsgamma}^*S_{d,\bsgamma}: H_{d,\bsgamma}\to
H_{d,\bsgamma} $ is a compact linear operator and the eigenpairs
of $W_{d,\bsgamma}$ are
$\{(\lz_{d,\bsgamma,\bsj},e_{\bsj})\}_{\bsj\in\Bbb N^d}$, where
$$\lz_{d,\bsgamma,\bsj}=\prod_{k=1}^d \lz_{k,j_k},  \ \ \  \ e_{\bsj}=e_{j_1}\otimes\cdots \otimes e_{j_d}.$$
$\lz_{k,j_k}:=\lz_{\gamma_k,j_k}$ is defined as in \eqref{1.1}.

In the following, we write $[d]:=\{1,2,\cdots,d\}.$ For $\bsj=(j_1,\cdots,j_d)\in \NN^d$, define
$$\uu(\bsj):=\{k \in [d] : j_k \ge 2\} \ \ \ \mbox{ and }\ \ \
\lambda_{d,\bsj}:= \lambda_{j_1}\lambda_{j_2}\cdots
\lambda_{j_d}.$$
The eigenvalues of $W_{d,\bsgamma}$ can also be written as
$$\lambda_{d,\bsgamma,\bsj}:= \gamma_{\uu(\bsj)} \lambda_{d,\bsj}=
\left(\prod_{k=1 \atop j_k \ge 2}^d \gamma_k\right)
\lambda_{j_1}\cdots \lambda_{j_d} = \left(\prod_{k\in \uu(\bsj)}
\gamma_k \right) \lambda_{d,\bsj}.$$
See \cite[Section~5.3]{NW08} or \cite{KPW}.

If $\gamma_j\equiv1$, the weighted tensor product space
$H_{d,\bsgamma}$  recedes to the tensor product space
$H_d=H_1\otimes \cdots \otimes H_1$, and the weighted linear
tensor product operator $S_{d,\bsgamma}$ reduces to the linear
tensor product operator $S_d=S_1\otimes \cdots \otimes S_1$.

Now consider  the linear weighted tensor product problem
$S_{\bsgamma}=\{S_{d,\bsgamma}\}_{d\in\NN}$. We want to
approximate $S_{d,\bsgamma}$ by algorithms $A_{d,n}(f)$ of the
form
\begin{equation}\label{1.2}
A_{d,n}(f)=\phi_{d,n}(L_1(f),L_2(f),\cdots,L_n(f)),
\end{equation}
where  $L_1,\cdots ,L_n$ belong to $\Lambda_d^{\rm
all}=H_{d,\bsgamma}^*$, the space consisting of all continuous
linear functionals on $H_{d,\bsgamma}$, and $\phi_{d,n}:\RR^n \to
G_d$ is an arbitrary  mapping. The choice of $L_j$ can be
adaptive, i.e., it may depend on the previously computed
information $L_1(f),L_2(f),\cdots,L_{j-1}(f)$. The worst case
error of an algorithm $A_{d,n}$ is defined as
$$ e(A_{d,n})=\sup_{f \in H_{d,\bsgamma} \atop \|f\|_{H_{d,\bsgamma}} \le 1}\|S_{d,\bsgamma}(f)-A_{d,n}(f)\|_{G_d}.$$
The $n$-th minimal worst-case error is defined by
$$e(n,S_{d,\bsgamma})= \inf_{A_{d,n}}e (A_{d,n}),$$ where the infimum is taken over all algorithms of the form \eqref{1.2}. When $n=0$, the initial error $e_0$ is denoted by
$$ e_0=\sup_{f \in H_{d,\bsgamma} \atop \|f\|_{H_{d,\bsgamma}} \le 1}\|S_{d,\bsgamma}(f)\|_{G_d}. $$

Clearly, we have
$$e_0=\|S_{d,\bsgamma}\|=\|W_{d,\bsgamma}\|^{1/2}=\max_{\bsj\in\NN^d}(\lz_{d,\bsgamma,\bsj})^{1/2}=1.$$ Hence, the
problem is well normalized for all $d\in\NN$ and all product
weights $\bsgamma$. The information complexity for the normalized
or absolute error criterion is defined by
$$n(\varepsilon,S_{d,\bsgamma})= \min\{n \, : \, e(n,S_{d,\bsgamma})\le \varepsilon\}.$$
It is known that for the class $\Lambda_d^{\rm all}$ the
information complexity is fully characterized in terms of the
singular values of $S_{d,\bsgamma}$, or equivalently, in terms of
the eigenvalues of $W_{d,\bsgamma}=S_{d,\bsgamma}^*S_{d,\bsgamma}$
(see \cite{NW08}).

Denote
$$A_{\varepsilon,d}=\{(n_1,\ldots,n_d)\in \NN^d :
\lambda_{1,n_1}\cdots \lambda_{d,n_d}> \varepsilon^2\}.$$ Then the
information complexity can be rewritten as
\begin{equation}\label{1.3}
n(\vz,S_{d,\bsgamma})=
 \bigg|\Big\{\bsj \in \NN^d : \lambda_{d,\bsgamma,\bsj}
 =\prod_{k=1}^d \lambda_{k,j_k}=\lambda_{1,j_1}\cdots \lambda_{d,j_d}>
 \varepsilon^2 \Big\}\bigg|=\big|A_{\vz,d}\big|,
\end{equation}
where $|A|$ denotes the number of the elements in a set $A$.
Clearly,
$$n(\vz,S_{d,\bsgamma})\le n(\vz_1,S_{d_1,\bsgamma})\ \ \ \mbox{for all $\vz_1\le \vz$ and $d_1\ge d$}.$$
Hence, for decreasing $\vz$ and increasing $d$, the information complexity
is non-increasing.

\subsection{Various  notions of EXP-tractability}
\
\newline
\indent In this subsection we briefly recall the basic
EXP-tractability notions. We are interested in the behavior of
$n(\varepsilon,S_{d,\bsgamma})$ when $d+\ln\vz^{-1}$ goes to
infinity in an arbitrary way. We define various notions of EXP
tractability.

The problem $S_{\bsgamma}=\{S_{d,\bsgamma}\}$ is said to be:

\begin{itemize}
 \item {\it Exponentially strongly polynomially
tractable (EXP-SPT)} if there exist $C, p \ge 0$ such that
 \begin{equation}\label{1.4}
 n(\varepsilon,S_{d,\bsgamma}) \le
C (1 + \ln  \varepsilon^{-1})^p, \ \ \ \ \forall d \in \NN,\
\forall \varepsilon \in (0,1].
 \end{equation}
 The infimum of all exponents $p \ge0$ such that
\eqref{1.4} holds for some $C \ge 0$ is called the exponent
of EXP-SPT and is denoted by $p^*$.
 \item {\it Exponentially polynomially tractable (EXP-PT)} if
there exist $C, p, q \ge 0$ such that
$$
n(\varepsilon,S_{d,\bsgamma}) \le C d^{\,q} (1 +
\ln\varepsilon^{-1})^p, \ \ \ \ \forall d \in \NN,\ \forall
\varepsilon \in (0,1].$$
 \item {\it Exponentially quasi-polynomially tractable
(EXP-QPT)} if there exist $C, t \ge 0$ such that
 \begin{equation}\label{1.5}
 n(\varepsilon,S_{d,\bsgamma})
\le C \exp(t(1+\ln d)(1+\ln(1 + \ln\varepsilon^{-1}))),\ \ \ \
\forall d \in \NN,\ \forall \varepsilon \in (0,1].
 \end{equation}
 The infimum of all exponents $t \ge0$ such that
\eqref{1.5} holds for some $C \ge 0$ is called the exponent
of EXP-QPT and is denoted by $t^*$.
 \item {\it Exponentially $(s,t)$-weakly tractable (EXP-$(s,t)$-WT)}
for fixed positive $s$ and $t$ if
$$
\lim_{d+\varepsilon^{-1} \rightarrow \infty} \frac{\ln
n(\varepsilon, S_{d,\bsgamma})}{d^t + (\ln\varepsilon^{-1})^s} = 0.
$$ If $s=t=1$, we speak of
{\it exponential weak tractability (EXP-WT)}. \item {\it
Exponentially uniformly weakly tractable (EXP-UWT)} if
EXP-$(s,t)$-WT holds for all positive $s$ and $t$.
\end{itemize}

\subsection{Main results}\label{Main Results}

\

 In the worst case setting, for the (unweighted) linear tensor
product  problem, necessary and sufficient conditions on
univariate singular values  under which we can or cannot achieve
different notions of exponential tractability were obtained in
\cite{ HKW19, PP14, PPW17}. However, we have many negative results
for EXP-tractability. For example, for the linear tensor product
problem, we cannot achieve EXP-WT.

For the weighted linear tensor product problems,  the authors in
\cite{KPW}  obtained  necessary and sufficient conditions for
various notions of EXP-tractability on weights $\{\gamma_j\}$ and
univariate singular values $\{(\lz_j)^{1/2}\}$, where in the
sequel we always assume that $\{\gamma_j\}$ is a non-increasing
positive sequence, and  $\{\lz_j\}$ is a  non-increasing sequence
satisfying $1=\lz_1\ge\lz_2>0$ and $\lim\limits_{j\to\infty}\lz_j=0$. See the following
EXP-tractability results of the problem $S_{\bsgamma}$.

\vskip 2mm

 $\bullet$  EXP-SPT holds  if and only if
$$
\lim_{j\to\infty}\lambda_j=\lim_{j\to\infty}\gamma_j=0\ \ \
\mbox{and}\ \ \ B_{{\rm EXP-SPT}}:=\limsup_{\varepsilon
\rightarrow 0} \frac{d(\varepsilon) \ln j(\varepsilon)}{\ln\ln
  \frac{1}{\varepsilon}}< \infty,$$where \begin{equation}
 j(\varepsilon) =  \max\{j \in \NN : \lambda_j >
 \varepsilon^2\}\ \ {\rm and}\ \ d(\varepsilon) =  \max\{d \in \NN : \gamma_d\,>\,\varepsilon^2\}.
\end{equation}
If this holds then the exponent of EXP-SPT is $p^*=B_{{\rm
EXP-SPT}}$.

 \vskip 3mm

 $\bullet$  EXP-SPT and EXP-PT are equivalent.  \vskip 3mm

 $\bullet$  EXP-QPT holds
 if and only if
$$
\lim_{j\to\infty}\lambda_j=\lim_{j\to\infty}\gamma_j=0\ \ \
\mbox{and}\ \ \ B_{{\rm EXP-QPT}}:=\limsup_{\varepsilon
\rightarrow 0} \frac{d(\varepsilon) \ln j(\varepsilon)}{[\ln
d(\varepsilon)] \ln\ln \frac{1}{\varepsilon}}< \infty.
$$
If this holds then the exponent of EXP-QPT is $t^*=B_{{\rm
EXP-QPT}}$.

 \vskip 3mm

 $\bullet$  Let $s=t=1$. EXP-WT holds  if and only if
$$\lim_{j \rightarrow \infty} \gamma_j=0\ \ \mbox{ and }\ \ \lim_{j \rightarrow \infty} \frac{\ln
\frac{1}{\lambda_j}}{\ln j}=\infty.$$
 \vskip 2mm

 $\bullet$  Let $s=1$  and $t<1$.
EXP-$(1,t)$-WT holds  if and only if  \begin{equation}
\label{1.7}\lim_{j \rightarrow \infty} \frac{\ln
\frac{1}{\gamma_j}}{\ln j}=\infty\ \ {\rm and }\ \ \lim_{j
\rightarrow \infty} \frac{\ln \frac{1}{\lambda_j}}{\ln j}=\infty.
\end{equation}
 \vskip 2mm

 $\bullet$  Let $s=1$ and $t>1$.
EXP-$(s,t)$-WT holds  if and only if
$$
\mbox{$\gamma_j$'s are arbitrary}\ \ \ \mbox{and}\ \ \
 \lim_{j \rightarrow \infty} \frac{\ln \frac{1}{\lambda_j}}{\ln
   j}=\infty.
$$
 \vskip 2mm

 $\bullet$  Let $s>1$, $t\le1$ and $\lambda_2<1$.
EXP-$(s,t)$-WT holds  if and only if
$$
\mbox{$\gamma_j$'s are arbitrary}\ \ \ \mbox{and}\ \ \ \lim_{j
\rightarrow \infty}\frac{(\ln\frac{1}{\lambda_j})^s}{\ln
  j}=\infty.
$$
 \vskip 2mm

 $\bullet$  Let $s>1$, $t\le1$ and $\lambda_2=1$.
EXP-$(s,t)$-WT holds  if and only if
$$
\exists\,p\in\NN \ \ \mbox{with}\ \ \gamma_p<1\ \ \ \mbox{and}\ \
\ \lim_{j \rightarrow
\infty}\frac{(\ln\frac{1}{\lambda_j})^s}{\ln
  j}=\infty.
$$
 \vskip 2mm

 $\bullet$ Let $s>1$ and $t>1$.
EXP-$(s,t)$-WT holds  if and only if
$$
\mbox{$\gamma_j$'s are arbitrary}\ \ \ \mbox{and}\ \ \ \lim_{j
\rightarrow \infty}\frac{(\ln\frac{1}{\lambda_j})^s}{\ln
  j}=\infty.
$$
 \vskip 2mm

 $\bullet$  Let $s<1$ and $t>1$.
EXP-$(s,t)$-WT holds for arbitrary $\gamma_j$'s  if and only if
\begin{equation}\label{1.8}
\lim_{j\to\infty}\frac{\left(\ln \frac1{\lambda_j}\right)^{\eta}}
{\ln j}=\infty \ \ \mbox{with}\ \ \eta=\frac{s(t-1)}{t-s}.
\end{equation}
 \vskip 2mm

 $\bullet$  Let $s<1$ and $t=1$.
EXP-$(s,1)$-WT holds  if and only if  for arbitrary integers
$d,k,j$ with $j\ge 2$ and $k\le d$ it is true that
\begin{equation}\label{1.9}
\lim_{d+\gamma_k^{-d}\lambda_j^{-d}\to\infty}
\frac{\left(\ln\frac1{\gamma_k}\right)^s+\left(\ln\frac1\lambda_j\right)^s
}{d^{1-s}\ln j}=\infty.
\end{equation}

\vskip 3mm

However, the authors in \cite{KPW} did not find out sufficient and
necessary conditions on $\{\gamma_j\}$ and $\{\lz_j\}$ for
EXP-$(s,t)$-WT with $\max(s,t)< 1$ and  EXP-UWT. In this paper, we
 solve the problem by filling the remaining gaps on
EXP-tractability. That is, we  obtain  necessary and sufficient
conditions for EXP-$(s, t)$-WT with $\max(s, t) < 1$ (see Theorem
\ref{thm1}) and for EXP-UWT (see Theorem \ref{thm2}).

\begin{thm}\label{thm1}
Let  $0<s<1$ and $0<t\leq1$. EXP-$(s,t)$-WT holds if and only if

\begin{equation}\label{1.10}
\lim_{\vz\to 0}\frac{\ln\ln j(\vz)}{\ln\ln\vz^{-1}}=0
\end{equation}
 and
\begin{equation}\label{1.11}
\lim_{\vz\to 0}\frac{d(\vz)^{1-s}\ln j(\vz)}{(\ln\vz^{-1})^s}=0.
\end{equation}
\end{thm}

\vskip 3mm

From Theorem \ref{thm1} we obtain the following corollary.

\begin{cor}\label{cor}
Let  $0<s<1$ and $0<t\leq1$.

(1) If EXP-$(s,t)$-WT holds, then we have
\begin{equation}\label{1.12}
 \lim_{j\rightarrow\infty}\frac{\ln(\ln\frac{1}{\lambda_j})}{\ln(\ln j)}=\infty
 \end{equation}
 and
\begin{equation}\label{1.13}
\lim_{j\rightarrow\infty}\frac{\ln\frac{1}{\gamma_j}}{j^\frac{1-s}{s}}=\infty;
\end{equation}

(2) if \eqref{1.12} holds and
\begin{equation}\label{1.14}
\lim_{j\rightarrow\infty}\frac{\ln\frac{1}{\gamma_j}}{j^{\frac{1-s}{s}+\delta}}=\infty
\end{equation}
holds for some   $\delta>0$, then EXP-$(s,t)$-WT holds.
\end{cor}
We remark that if $\lz_{j}=0$ for some $j\in\NN$, then \eqref{1.12} holds.

\begin{thm}\label{thm2}   EXP-UWT holds if and only if \eqref{1.12} holds and
 \begin{equation}\label{1.15}
\lim_{j\rightarrow\infty}\frac{\ln(\ln\frac{1}{\gamma_j})}{\ln j}=\infty.
\end{equation}
\end{thm}

The paper is organized as follows. In Section 2 we give some
necessary preliminaries in the worst case setting. In Section 3,
we give the proofs of Theorems \ref{thm1} and \ref{thm2}.

\section{Premilinaries}\

For $\varepsilon \in (0,1)$ and
$\lim\limits_{j\to\infty}\lambda_j=0$, define
\begin{equation}
 j(\varepsilon) =  \max\{j \in \NN : \lambda_j > \varepsilon^2\}.
\end{equation}
Then $j(\vz)$ is well defined and always finite. We put $j(\vz)=0$
for $\vz\ge1$.  Since $\lambda_1=1$, we have $j(\vz)\ge1$ for
$\vz\in(0,1)$. Furthermore, $j(\vz)$ goes
to infinity if and only if all $\lambda_j$'s are positive.

We also assume that the weights satisfy
\begin{equation}\label{2.2}
1\ge \gamma_1 \ge \gamma_2 \ge \gamma_3 \ge \dots > 0.
\end{equation}
For $\varepsilon\in(0,1)$, define
$$d(\varepsilon) =  \sup\{d \in \NN : \gamma_d\,>\,\varepsilon^2\}.$$
Then $d(\vz)$ is well defined. We put $d(\vz)=0$ for
$\gamma_1\le\vz^2$, and note that $d(\vz)\ge1$ for
$\gamma_1>\vz^2$. Both~$j(\vz)$ and $d(\vz)$ are non-decreasing,
and $\lim\limits_{\varepsilon\to0}d(\varepsilon)=\infty$. If
$\lim\limits_{j\to\infty}\gamma_j=0$, then $d(\vz)<\infty$ for any
$\vz>0$.

The following lemma shows how the information complexity can be
bounded in terms of $j(\varepsilon)$ and $d(\varepsilon)$.

\begin{lem} {\rm (See \cite[Lemma 1]{KPW})}\label{lem1}
If
$\lim\limits_{j\to\infty}\lambda_j=\lim\limits_{j\to\infty}\gamma_j=0$,
then for $\varepsilon\in (0,1)$, we have
$$
n(\varepsilon,S_{d,\bsgamma}) \le
j(\varepsilon)^{\min(d,d(\varepsilon))}\le n(\varepsilon^{2d(\vz)},S_{d,\bsgamma}),
$$
and for $d\ge d(\vz)$,
$$n(\varepsilon,S_{d,\bsgamma}) =n(\varepsilon,S_{d(\vz),\bsgamma}).$$
\end{lem}

\begin{lem}\label{lem2}
 Let $\{\lambda_j\}_{j\in\NN}$ be a non-increasing sequence and $\lim\limits_{j\rightarrow\infty}\lambda_j=0$. Then \eqref{1.10} holds  if and only if   \eqref{1.12} holds  if and only if
\begin{equation}\label{2.3}
\lim_{\vz\to 0}\frac{\ln j(\vz)}{(\ln {\vz}^{-1})^\delta}=0
\end{equation}
holds for any $\delta>0$.
\end{lem}

\begin{proof}If $\lz_{j}=0$ for
some $j_0\in\NN$, then $\lambda_j=0$ for $j\ge j_0$ and hence
$j(\vz)\le j_0$. It follows that \eqref{1.10}, \eqref{1.12} and \eqref{2.3} hold in this case. Without loss of generality, we
assume that $\{\lz_j\}$ is a positive sequence.

We first show that \eqref{1.10} implies \eqref{1.12}. Assume that
\eqref{1.10} holds. For every $\eta\in(0,1)$, there exists an
$\vz_0\in(0,1)$ such that
$$\frac{\ln\ln j(\vz)}{\ln\ln \vz^{-1}}<\eta  \ \ \ {\rm for}\  \vz\in(0,\vz_0).$$
This implies
$$j(\vz)<e^{ (\ln\vz^{-1})^\eta}\le \lceil e^{ (\ln\vz^{-1})^\eta}\rceil :=v(\vz),$$
where $\lceil x\rceil$  is the ceiling of a real number $x$.  From the definition of $j(\vz)$, we
have $j(\vz)+1\le v(\vz)$ and
$$\lambda_{v(\vz)}\le \lambda_{j(\vz)+1}\le \vz^2.$$
Let $v(\vz)=k$. If we vary $\vz\in(0,\vz_0)$, then $k$ can take the
values
$$k= \lceil e^{ (\ln\vz_0^{-1})^\eta}\rceil,\lceil e^{ (\ln\vz_0^{-1})^\eta}\rceil+1,\cdots.$$
Since $k\le e^{(\ln\vz^{-1})^\eta}+1$, we get
$$\vz^{2}\le e^{-2\big(\ln (k-1)\big)^{1/\eta}}.$$
Hence, $$\lambda_{k}\le \vz^2\le e^{-2\big(\ln(k-1)\big)^{1/\eta}}\ \  \ {\rm for}\  k\ge \lceil e^{(\ln\vz_0^{-1})^\eta}\rceil. $$
 Taking the logarithm yields
  $$\frac{\ln\ln \frac{1}{\lambda_{k}}}{\ln\ln k}\ge \frac{\ln2+\frac{1}{\eta}\ln\ln(k-1)}{\ln\ln k}\to \infty
  \ \ \ {\rm as \ }k\to\infty, \ \eta\to 0.$$
Therefore, we obtain \eqref{1.12}.

Next, if $\eqref{1.12}$ holds, then for every $\delta\in(0,s)$ there exists a
$J>0$ such that
$$\ln\big(\ln\frac{1}{\lambda_j}\big)\geq \frac{2}{\delta}\ln(\ln j) \ \ \ {\rm for}\ j\geq J.$$
This implies $\lambda_j\le e^{-(\ln j)^{\frac{2}{\delta}}}$. We solve
the inequality $\vz^2\ge e^{-(\ln j)^{\frac{2}{\delta}}},$ and obtain the solution
$$j\ge e^{(2\ln\vz^{-1})^{\frac{\delta}{2}}}.$$
This concludes that if $$j\ge \max\{J,\ e^{(2\ln\vz^{-1})^{\frac{\delta}{2}}}\},$$
then $$\lz_j\le e^{-(\ln j)^{\frac{2}{\delta}}}\le \vz^2.$$
 From the definition of $j(\vz)$, we obtain
$$j(\vz)\le \max\{J,\ e^{(2\ln\vz^{-1})^{\frac{\delta}{2}}}\}.$$
Taking the logarithm yields that for every $\delta\in(0,s)$,
$$\lim_{\vz\to 0}\frac{\ln j(\vz)}{(\ln \vz^{-1})^s}\le \lim_{\vz\to 0}\frac{\max\{\ln J,\ (2\ln\vz^{-1})^{\frac{\delta}{2}}\}}{(\ln \vz^{-1})^\delta}=0.$$
Therefore, \eqref{2.3} holds.

Now assume that \eqref{2.3} holds, we want to show \eqref{1.10}.
From \eqref{2.3}, there exists an $\vz_0>0$ such that
$$\ln j(\vz)< \big(\ln {\vz}^{-1}\big)^\delta \ \ \ {\rm for }\  \vz\in (0,\vz_0).$$
Taking the logarithm we have
$$\frac{\ln\ln j(\vz)}{\ln\ln{\vz}^{-1} }< \delta.$$
Hence, \eqref{1.10} follows by letting $\delta\rightarrow0$.
Lemma \ref{lem2} is proved.
\end{proof}\vskip 1pc

\begin{lem}\label{lem3}Let $\{\gamma_j\}_{j\in\NN}$ be a non-increasing positive
sequence and $\lim\limits_{j\rightarrow\infty}\gamma_j=0$.
\begin{equation}\label{2.4}
\lim_{\vz\to 0}\frac{d(\vz)^{1-s}}{(\ln{\vz}^{-1})^s}=0
\end{equation}
 holds if and only if
\begin{equation}\label{2.5}
\lim_{j\rightarrow\infty}\frac{\big(\ln\frac{1}{\gamma_j}\big)^s}{j^{1-s}}=\infty.
\end{equation}

\end{lem}
\begin{proof}
Assume that \eqref{2.5} holds.  For every $\delta>0$, there exists a $j_0\in\NN$ such that
$$\big(\ln\frac{1}{\gamma_j}\big)^s>\frac{2}{\delta}\cdot j^{1-s} \ \ \ {\rm  for\ all}\ j\ge j_0.$$
 Notice that $\lim\limits_{\vz\to0}d(\vz)=\infty$. For the above $j_0$ and sufficiently small $\vz>0$, we have
$d(\vz)>j_0$. This yields that
$$\big(\ln\frac{1}{\gamma_{d(\vz)}}\big)^s>\frac{2}{\delta}\cdot d(\vz)^{1-s}.$$
Since $\gamma_{d(\vz)}> \vz^2$, we have
$$\ln\frac{1}{\gamma_{d(\vz)}}< \ln\vz^{-2}=2\ln\vz^{-1}.$$
Hence,
$$\frac{2}{\delta}\cdot d(\vz)^{1-s}<\big(\ln\frac{1}{\gamma_{d(\vz)}}\big)^s<\big(2\ln\vz^{-1}\big)^s.$$
This implies
$$\frac{d(\vz)^{1-s}}{\big(\ln\vz^{-1}\big)^s}<2^{s-1}\delta.$$
Since  $\delta$ can be arbitrary small, \eqref{2.4} follows.

Next, assume that \eqref{2.4} holds. For every $\delta>0$, there exists an $\vz_0\in(0,1)$ such that for $0<\vz\le\vz_0$,
$$\frac{d(\vz)^{1-s}}{(\ln\vz^{-1})^s}<\delta.$$
This implies
$$d(\vz)<\delta^{\frac1{1-s}}\big(\ln\vz^{-1}\big)^{\frac s{1-s}}\le \big\lceil \delta^{\frac1{1-s}}
\big(\ln\vz^{-1}\big)^{\frac s{1-s}}\big\rceil :=v(\vz).$$
 Then $d(\vz)+1\le v(\vz).$ From the definition of $d(\vz)$, we have
$$\gamma_{v(\vz)}\le \gamma_{d(\vz)+1}\le \vz^2.$$
Let $v(\vz)=k$, if we vary $\vz\in(0,\vz_0)$ then $k$ can take the
values
$$k= \big\lceil \delta^{\frac 1{1-s}}\big(\ln\vz_0^{-1}\big)^{\frac s{1-s}}\big\rceil,
\big\lceil\delta^{\frac1{1-s}}\big(\ln\vz_0^{-1}\big)^{\frac
s{1-s}}\big\rceil+1,\cdots.$$ Since
$k\le\delta^{\frac1{1-s}}\big(\ln\vz^{-1}\big)^{\frac s{1-s}}+1$,
we get
$$\vz^{2}\le e^{-2\big({k-1}\big)^{\frac {1-s}s}\delta^{-1/s}}.$$
Then
$$\gamma_{k}\le \vz^2\le  e^{-2\big({k-1}\big)^{\frac {1-s}s}\delta^{-1/s}}$$
for $k\ge \big\lceil\delta^{\frac
1{1-s}}\big(\ln\vz_0^{-1}\big)^{\frac s{1-s}}\big\rceil.$ Taking
the logarithm we obtain
 $$\frac{\big(\ln\frac{1}{\gamma_{k}}\big)^s}{k^{1-s}}\ge \frac{2^s}{\delta}\cdot\frac{(k-1)^{1-s}}{k^{1-s}}\to
 \frac{2^s}{\delta}\ \ \ {\rm as \ }k\to\infty.$$
 Since $\delta$ can be sufficiently small, we obtain \eqref{2.5}. Lemma \ref{lem3} is proved.
\end{proof}

\begin{lem}\label{lem4}
 Let $\{\lambda_j\}_{j\in\NN}$ be a non-increasing sequence and $\lim\limits_{j\rightarrow\infty}\lambda_j=0$. Then \eqref{1.12} holds if and only if for any
$\delta>0$,
\begin{equation}\label{2.6}
\lim_{j\rightarrow\infty}\frac{(\ln\frac{1}{\lambda_j})^\alpha}{\ln
j}=\infty\ \ \ \ {\rm for\ all}\ \alpha\in(0,\delta).
\end{equation}
\end{lem}
\begin{proof}
Since $\lambda_j$ is non-increasing and
$\lim\limits_{j\rightarrow\infty}\lambda_j=0$. If $\lz_{j}=0$ for
some $j_0\in\NN$, then for $j\ge j_0$, we have $\ln\frac{1}{\lambda_j}=\infty$, and hence both \eqref{2.6} and
\eqref{1.12} hold.  Without loss of generality  we assume that
$\{\lz_j\}$ is a positive sequence.

 Assume that $\eqref{2.6}$ holds. There exists a $J_1>0$ such that
$$\big(\ln \frac{1}{\lambda_j}\big)^\alpha\geq \ln j \ \ \ {\rm for }\  j\geq J_1.$$
Taking the logarithm, we have
$$\frac{\ln(\ln\frac{1}{\lambda_j})}{\ln(\ln j)}\geq\frac{1}{\alpha}.$$
Now $\eqref{1.12}$ follows by letting $\alpha\rightarrow0$.

On the other hand, if $\eqref{1.12}$ holds, then for
any $\alpha\in(0,\delta)$, there exists a $J_2$ such that
$$\ln(\ln\frac{1}{\lambda_j})\geq \frac 2\alpha \ln(\ln j) \ \ \ {\rm for}\ j\geq J_2.$$
This implies
$$\frac{(\ln\frac{1}{\lambda_j})^{\alpha}}{\ln j}\geq\ln j \ \ \ {\rm for} \ j\geq J_2.$$
Inequality $\eqref{2.6}$ follows as $j\to\infty$.
\end{proof}

Using the same method, we can obtain the following lemma.

\begin{lem}\label{lem5}
\indent Let $\{\gamma_j\}_{j\in\NN}$ be a non-increasing positive
sequence and $\lim\limits_{j\rightarrow\infty}\gamma_j=0$. Then
\eqref{1.15} holds if and only if for any $\delta>0$,
$$\lim_{j\rightarrow\infty}\frac{(\ln\frac{1}{\gamma_j})^\alpha}{j}=\infty\ \ \ \ {\rm for\ all}\ \alpha\in(0,\delta).$$
\end{lem}\vskip 1pc

Next, we partition the set $[d]$. Without loss of
generality, we assume that $d\ge8$. Set
$$\rho(i)=\big\{k\in\NN: 2^{i-1}\le k \le2^i-1\big\}, \ \ \ {\rm for} \ i=1,\cdots,\lfloor\log_2d\rfloor-1,$$
and
$$\rho\big(\lfloor\log_2d\rfloor\big)=\big\{k\in\NN:2^{\lfloor\log_2d\rfloor-1}\le k\le d\big\}.$$
Then
\begin{equation}\label{2.7}
[d]=\bigcup_{i=1}^{\lfloor\log_2d\rfloor}\rho(i).
\end{equation}
We have
$$2^{i-1}\le|\rho(i)|<2^{i+1}, \ i=1,\cdots,\lfloor\log_2d\rfloor,$$
where $|A|$ denotes the number of the elements in set $A$.
 Let $C_n^k$ represent the number of combinations of $k$ elements selected from $n$ elements, i.e.,
 $$C_n^k=\frac{n!}{k!(n-k)!}.$$
For $\vz\in(0,1)$, set
$$A_{\rho(i)}(\vz):=\big\{\bsj\in\NN^{|\rho(i)|}:
\prod\limits_{k\in\rho(i)}\lambda_{2^{i-1},j_k}>\vz^2\big\}.$$

\begin{lem}\label{lem6}  We have
\begin{equation}\label{2.8}
\big|A_{\rho(i)}(\vz)\big|
\le C_{|\rho(i)|}^{a_i(\vz)}\, j\Big(\frac{\vz}{\sqrt{\gamma_{2^{i-1}}}}\Big)^{a_i(\vz)},
\end{equation}
where $a_i(\vz):=\min\{|\rho(i)|, \lceil\frac{\ln\vz^{-2}}{\ln(\gamma_{2^{i-1}}\lambda_2)^{-1}}\rceil-1\}.$
\end{lem}

\begin{proof} For $\bsj\in A_{\rho(i)}(\vz)$,  set $m=|\{k\in\rho(i) : j_k\ge 2\}|$. Then
$$(\gamma_{2^{i-1}}\cdot \lambda_2)^m\ge\prod_{k\in\rho(i)}\lambda_{2^{i-1},j_k}>\vz^2.$$
It implies that
$$m\le \big\lceil\frac{\ln\vz^{-2}}{\ln(\gamma_{2^{i-1}}\lambda_2)^{-1}}\big\rceil-1.$$
Then we have
\begin{equation*}
m\le \min\big\{|\rho(i)|,
\big\lceil\frac{\ln\vz^{-2}}{\ln(\gamma_{2^{i-1}}\lambda_2)^{-1}}\big\rceil-1\big\}=a_i(\vz),
\end{equation*}
i.e., there are at most $a_i(\vz)$ indices $j_k$ which satisfy
$j_k\ge 2$.  For sufficiently small $\vz>0$, if there exists a
$j_k>j\Big(\frac{\vz}{\sqrt{\gamma_{2^{i-1}}}}\Big)$ for some
$k\in\rho(i)$, then we get
 $$\prod_{k\in\rho(i)}\lambda_{2^{i-1},j_k}\le \gamma_{2^{i-1}}\cdot
 \lambda_{j_k}\le\vz^2,$$which means that $\bsj\not\in
 A_{\rho(i)}(\vz)$.
Hence,  for $\bsj\in A_{\rho(i)}(\vz)$ there are at most
\begin{equation*}
C_{|\rho(i)|}^{|\rho(i)|-a_i(\vz)}=C_{|\rho(i)|}^{a_i(\vz)}
=\frac{|\rho(i)|!}{a_i(\vz)!\big(|\rho(i)|-a_i(\vz)\big)!}
\end{equation*}
ways to select the $|\rho(i)|-a_i(\vz)$ indices $j_k$ that must be
equal to 1, and in the remaining $a_i(\vz)$ indices $j_k$ of
$\rho(i)$ only
$$j_k\in\Big\{1, 2, \cdots, j\Big(\frac{\vz}{\sqrt{\gamma_{2^{i-1}}}}\Big)\Big\}$$ may belong to $A_{\rho(i)}(\vz).$
Then \eqref{2.8} follows. Lemma \ref{lem6} is proved.
\end{proof}

Furthermore, we have the more delicate estimate. However, we do
not use it in this paper.
\begin{lem}\label{lem7}
We have
\begin{equation}\label{2.9}
\big|A_{\rho(i)}(\vz)\big|\le 1+ \sum_{m=1}^{a_i(\vz)}\frac{|\rho(\vz)|!}{(|\rho(\vz)|-m)!}\prod_{k=1}^{m}j\bigg(\Big(\frac{\vz}{\sqrt{(\gamma_{2^{i-1}})^m\lambda_2^{m-k}}}\Big)^{\frac{1}{k}}\bigg).
\end{equation}
\end{lem}
\begin{proof} For $\bsj\in A_{\rho(i)}(\vz)$,  set $m_{\bsj}=|\{k\in\rho(i) : j_k\ge 2\}|$.
It follows from  the proof of Lemma \ref{lem6} that
 $$m_{\bsj}\le a_i(\vz):=\min\{|\rho(i)|,
\lceil\frac{\ln\vz^{-2}}{\ln(\gamma_{2^{i-1}}\lambda_2)^{-1}}\rceil-1\}.$$For
$1\le m\le a_i(\vz)$, we put $$  A_{\rho(i)}^m(\vz)=\{\bsj\in
\NN^{|\rho(i)|} : \bsj\in A_{\rho(i)}(\vz) \ {\rm and}\
m_{\bsj}=m\}.$$
 Clearly,
$$\big|A_{\rho(i)}(\vz)\big|=1+\sum_{m=1}^{a_i(\vz)} \big|A_{\rho(i)}^m(\vz)\big|.$$

For $\bsj\in A_{\rho(i)}^m(\vz)$,  let $j_{1,max}$ be the largest
index of the eigenvalues in the product
$\prod\limits_{k\in\rho(i)}\lambda_{2^{i-1},j_k}.$ Since
$$(\gamma_{2^{i-1}})^m \cdot \lambda_{j_{1,max}}\lambda_2^{m-1} \ge \prod_{k\in\rho(i)}\lambda_{2^{i-1},j_k}>\vz^2,$$
we obtain that
$$j_{1,max}\le \max\Big\{j:\lambda_j>\frac{\vz^2}{(\gamma_{2^{i-1}})^m\lambda_2^{m-1}}\Big\}
=j\Big(\frac{\vz}{{(\gamma_{2^{i-1}})^{m/2}\lambda_2^{(m-1)/2}}}\Big).$$
Let $j_{k,max}$ be the $k$-th largest index of the eigenvalues in
the product $\prod\limits_{k\in\rho(i)}\lambda_{2^{i-1},j_k},$
$k=2,\cdots,m$. Since
$$(\gamma_{2^{i-1}})^m \cdot \lambda_{j_{k,max}}^k\lambda_2^{m-k}\ge  \prod_{k\in\rho(i)}\lambda_{2^{i-1},j_k}
>\vz^2,$$
we obtain
\begin{equation*}
j_{k,max}\le
\max\Big\{j:\lambda_j>\Big(\frac{\vz^2}{(\gamma_{2^{i-1}})^m\lambda_2^{m-k}}\Big)^{\frac{1}{k}}\Big\}
=
j\bigg(\Big(\frac{\vz}{{(\gamma_{2^{i-1}})^{m/2}\lambda_2^{m/2-k/2}}}\Big)^{\frac{1}{k}}\bigg).
\end{equation*}

Hence, for $\bsj\in A_{\rho(i)}^m(\vz)$, there are
$C_{|\rho(i)|}^{m}$ ways to select the $|\rho(i)|-m$ indices that
must be equal to 1, and the remaining indices that must be greater
than 1.  For $k=1,2,\cdots,m$, only  the $k$-th largest indices
$$j_{k,max}\in\bigg\{2, \cdots, j\bigg(\Big(\frac{\vz}
{{(\gamma_{2^{i-1}})^{m/2}\lambda_2^{m/2-k/2}}}\Big)^{\frac{1}{k}}\bigg)\bigg\},\
$$ may belong to $A_{\rho(i)}^m(\vz)$. And there are $m!$
arrangements for the remaining $m$ indices. Therefore, we obtain
$\eqref{2.9}$.
\end{proof}

\begin{lem}\label{lem8}Let $C_m^s$ be the combination number. Then we have
\begin{equation}\label{2.10}
\ln C_m^s\le s\ln\frac{em}{s}.
\end{equation}
\end{lem}
\begin{proof}
From the Stirling's formula, we have
$$s!\ge \big(\frac{s}{e}\big)^s.$$
It follows that
$$C_m^s=\frac{m!}{s!(m-s)!}\le \frac{m^s}{s!}\le \frac{m^s}{(\frac{s}{e})^s}=\big(\frac{em}{s}\big)^s.$$
Taking the logarithm, we get \eqref{2.10}.
\end{proof}

\section{Proofs of Theorems \ref{thm1} and \ref{thm2}}

\

\noindent{\it Proof of Theorem \ref{thm1}.}

 Assume that EXP-$(s,t)$-WT holds for
fixed $s<1$ and $t\leq1$. We want to show that \eqref{1.10} and \eqref{1.11} holds. We know that EXP-$(s,t)$-WT  holds also for
this $s$ and any $t_1>1$. From \eqref{1.8} we obtain that
\begin{equation}\label{3.1}\lim_{j\rightarrow\infty}\frac{(\ln\frac{1}{\lambda_j})^{s\frac{t_1-1}{t_1-s}}}{\ln
j}=\infty \ \ \ {\rm for \ any}\ t_1>1.\end{equation} Obviously,
$\alpha=s\frac{t_1-1}{t_1-s}\in(0,s)$. Applying \eqref{3.1} and
Lemma $\ref{lem4}$ with $\alpha\in (0,s)$, we obtain \eqref{1.12}. By Lemma \ref{lem2}, we get \eqref{1.10}.

Now we  show \eqref{1.11}. Since  EXP-$(s,t)$-WT holds with $s<1$
and $t\le 1$, then EXP-WT also holds and hence
$\lim\limits_{j\to\infty}\gamma_j=0$. Applying Lemma \ref{lem1}
with $d=d(\vz)$, we have
$$j(\vz)^{d(\vz)}\le n(\vz^{2d(\vz)},S_{d(\vz),\gamma}).$$
It follows that
$$\frac{d(\vz)\ln j(\vz)}{d(\vz)+(2d(\vz)\ln\vz^{-1})^s}\le\frac{\ln n(\vz^{2d(\vz)},
S_{d(\vz),\gamma})}{d(\vz)^t+(2d(\vz)\ln\vz^{-1})^s}\to 0\ \ \
{\rm as}\ \vz\to 0.$$ This leads to
$$\lim_{\vz\to 0^+}\frac{d(\vz)\ln j(\vz)}{d(\vz)+(2d(\vz)\ln\vz^{-1})^s}=\lim_{\vz\to 0^+}\frac{\ln j(\vz)}{1+2^sd(\vz)^{s-1}(\ln\vz^{-1})^s}=0.$$
Since $\ln j(\vz)\ge \ln2$ for $0<\vz<\lz_2$, it requires that
$$\lim_{\vz\to 0^+} d(\vz)^{s-1}(\ln\vz^{-1})^s=\infty.$$
From this we conclude that
$$\lim_{\vz\to 0^+}\frac{d(\vz)^{1-s}}{(\ln\vz^{-1})^{s}}=0 \ \
{\rm or,\ equivalently,}\ \ \
d(\vz)^{1-s}=o(1)(\ln\vz^{-1})^s.$$
Hence
\begin{equation*}
\lim_{\vz\to 0^+}\frac{d(\vz)\ln j(\vz)}{d(\vz)+2^sd(\vz)^s(\ln\vz^{-1})^s}
=\lim_{\vz\to 0^+}\frac{d(\vz)^{1-s}\ln
j(\vz)}{(o(1)+2^s)(\ln\vz^{-1})^s}=0.
\end{equation*}
Therefore \eqref{1.11} holds.

On the other hand, assume that \eqref{1.10} and \eqref{1.11} hold.
We want to show that EXP-$(s,t)$-WT holds for this $s$ and $t$. By Lemma \ref{lem2}, we get \eqref{1.12}. From \eqref{1.11} and the fact that $\ln j(\vz)\ge \ln2$ for $\vz\in
(0,\lz_2)$, we get \eqref{2.4}. According to Lemma \ref{lem3}, it is equivalent to
\eqref{2.5}. Hence
\begin{equation}\label{3.2}
\lim_{j\to\infty}\frac{\ln \frac{1}{\gamma_j\lambda_2}}{j^{\frac{1-s}{s}}}=\infty.
\end{equation}

Set $$\ln
\frac{1}{{\gamma}_k\lambda_2}=k^{\frac{1-s}{s}}\widehat{h}(k),\ \
{\rm and}\ \ \widetilde{h}(k)=\inf\limits_{j\geq
k}\widehat{h}(j).$$
 Then the sequence $\{\widetilde{h}(k)\}_{k\in\NN}$ is non-decreasing, and  satisfies
 $$\widehat{h}(k)\geq \widetilde{h}(k) \ \ {\rm and} \ \ \lim_{k\rightarrow\infty}\widetilde{h}(k)=\lim_{k\rightarrow\infty}\widehat{h}(k)=\infty.$$
We put
$$h(1)=\widetilde{h}(1),\ \  h(k+1)=\min\Big\{\,\frac{\ln(k+2)}{\ln(k+1)}h(k),\ \widetilde{h}(k+1)\,\Big\},\ k=1,2,\cdots.$$
Clearly, it satisfies
$$h(k)\le\frac{\ln(k+2)}{\ln(k+1)}h(k)\ \ \ {\rm and}\ \ \   h(k)\le \widetilde{h}(k)\le\widetilde{h}(k+1),$$
 which yields
 $$h(k)\le \min\Big\{\,\frac{\ln(k+2)}{\ln(k+1)}h(k),\widetilde{h}(k+1)\,\Big\}=h(k+1).$$
Then $\{h(k)\}_{k\in\NN}$ is a non-decreasing sequence,
$$\lim_{k\to\infty}h(k)=\infty,\ \ \ {\rm and}\ \   \widehat{h}(k)\ge \widetilde{h}(k)\ge h(k).$$
We also note that for $m>n$, \begin{equation}\label{3.3}h(m)\le
\frac{\ln(m+1)}{\ln m}h(m-1)\le \cdots\le
\frac{\ln(m+1)}{\ln(n+1)}h(n),\end{equation} and
$$h(2n)\le 2h(n) \ \ {\rm for}\  n\ge1.$$
Put
$$\widetilde{\lz}_{k,j_k}=\left\{
\begin{array}{ll}
1, & \mbox{ if } j_k=1,\\
\widetilde{\gamma}_k \lambda_{j_k}, & \mbox{ if } j_k \ge 2,
\end{array}\right.  \ \ {\rm where}\ \
\ln\frac{1}{\widetilde{\gamma}_k\lambda_2} =k^{\frac{1-s}s}h(k),\ k=1,2,\cdots.$$

If $\ln\vz^{-1}\le d^t$, then $d\to\infty$ as
$\vz^{-1}+d\to\infty$. From \eqref{1.12}, \eqref{3.2} and  Lemma \ref{lem4}, we
obtain \eqref{1.7}. Then EXP-$(1,t)$-WT holds for this $t$. We
have
$$0=\lim_{\vz^{-1}+d\to\infty}\frac{\ln n(\vz,S_{d,\bsgamma})}{\ln\vz^{-1}+d^t}\ge
\lim_{d\to\infty}\frac{\ln n(\vz,S_{d,\bsgamma})}{2d^t}\ge0 .$$
It follows that
$$\lim_{\vz^{-1}+d\to\infty}\frac{\ln n(\vz,S_{d,\bsgamma})}{(\ln\vz^{-1})^s+d^t}\le \lim_{d\to\infty}\frac{\ln n(\vz,S_{d,\bsgamma})}{d^t}=0.$$
In this case, EXP-$(s,t)$-WT holds.

If $\ln\vz^{-1}\ge d^t$, then $\vz^{-1}\to \infty$ as
$\vz^{-1}+d\to\infty$. We use \eqref{1.3} and \eqref{2.7} to obtain
\begin{align*}
n(\vz,S_{d,\gamma})&=\big|\{\bsj\in\NN^d : \prod_{k=1}^{d}\lambda_{k,j_k}>\vz^2\}\big|\\
&\le \prod_{i=1}^{\lfloor\log_2d\rfloor}\big|\{\bsj\in\NN^{|\rho(i)|} :\prod_{k\in\rho(i)}\lambda_{k,j_k}>\vz^2 \}\big|\\
&\le \prod_{i=1}^{\lfloor\log_2d\rfloor}\big|\{\bsj\in\NN^{|\rho(i)|} :\prod_{k\in\rho(i)}\lambda_{2^{i-1},j_k}>\vz^2 \}\big|\\
&\le \prod_{i=1}^{\lfloor\log_2d\rfloor}\big|\{\bsj\in\NN^{|\rho(i)|} :\prod_{k\in\rho(i)}\widetilde{\lambda}_{2^{i-1},j_k}>\vz^2 \}\big|.
\end{align*}
Taking the logarithm, we use Lemma \ref{lem6} to obtain
\begin{align}\label{3.4}
\ln n(\vz,S_{d,\gamma})
&\le \sum_{i=1}^{\lfloor\log_2d\rfloor}\Big(\ln C_{|\rho(i)|}^{\widetilde{a}_i(\vz)}+\widetilde{a}_i(\vz)\ln j(\vz)\Big),
\end{align}
where $\widetilde{a}_i(\vz)=\min\{|\rho(i)|,\,\lceil\frac{\ln\vz^{-2}}{\ln(\widetilde{\gamma}_{2^{i-1}}\lambda_2)^{-1}}\rceil-1\}$.
We have
$$\widetilde{a}_i(\vz)\le \min\{2^{i+1},\,\frac{\ln\vz^{-2}}{\ln(\widetilde{\gamma}_{2^{i-1}}\lambda_2)^{-1}}\}\le \min\{2^{i+1},\,\frac{\ln\vz^{-2}}{2^{(i-1)\frac{1-s}{s}}h(2^{i-1})}\}.$$
Note that $\{2^{l+2}\,2^{l\cdot\frac{1-s}s}h(2^l)\}$ is
increasing. Choose $l_0$ such that
$$2^{l_0+1}\cdot 2^{(l_0-1)\frac{1-s}{s}}h(2^{l_0-1})\le \ln\vz^{-2}<2^{l_0+2}\cdot 2^{l_0\cdot\frac{1-s}{s}}h(2^{l_0}).$$
This implies
\begin{equation}\label{3.5}
2^{l_0}\asymp \bigg(\frac{\ln\vz^{-2}}{h(2^{l_0})}\bigg)^s \ \ {\rm or,\ equivalently, }\  \ \ \ln\vz^{-2}\asymp 2^{l_0/s}h(2^{l_0}),
\end{equation}
where the notation $A\lesssim
B$ means that there exists a nonessential positive constant $C$
such that $A\le C B$, and $A\asymp B$ means that $A\lesssim B$ and
$B\lesssim A$. By \eqref{2.10} we note that if $1\le i\le l_0$,
then
$$
\ln C_{|\rho(i)|}^{\widetilde{a}_i(\vz)}\le 2^i, $$and if $i>
l_0$, then by \eqref{3.3} and \eqref{3.5} we have
\begin{align*} \ln
\frac{e\,|\rho(i)|}{\widetilde{a}_i(\vz)}& \lesssim 1+\ln
\Big(\frac{2^{i+1}\cdot 2^{(i-1)\frac{1-s}{s}}h(2^{i-1})}{
\ln\vz^{-2}}\Big)\lesssim
1+\ln\Big( 2^{\frac{i-l_0}{s}}\cdot \frac{h(2^{i-1})}{h(2^{l_0})}\Big)\\
&\lesssim 1+
\frac{i-l_0}{s}+\ln\Big(\frac{\ln(1+2^{i-1})}{\ln(1+2^{l_0})}\Big)\lesssim i-l_0+
\ln \frac{e\,i}{l_0},
\end{align*}
and hence
\begin{align*}\ln C_{|\rho(i)|}^{\widetilde{a}_i(\vz)} &\le
\widetilde{a}_i(\vz) \ln
\frac{e\,|\rho(i)|}{\widetilde{a}_i(\vz)}\lesssim
\frac{\ln\vz^{-2}}{2^{(i-1)\frac{1-s}{s}}h(2^{i-1})}\ln
\frac{e\,|\rho(i)|}{\widetilde{a}_i(\vz)}\\&\lesssim
\frac{\ln\vz^{-2}}{2^{(i-1)\frac{1-s}{s}}h(2^{l_0})}\Big(i-l_0+\ln
\frac{e\,i}{l_0}\Big).
\end{align*}
Going back to \eqref{3.4} we obtain
\begin{align*}
 &\quad\ \ln n(\vz, S_{d, \bsgamma})\le \sum_{i=1}^{\infty}\Big(\ln C_{|\rho(i)|}^{\widetilde{a}_i(\vz)}+\widetilde{a}_i(\vz)\ln j(\vz)\Big)\\
 &\lesssim\sum_{i=1}^{l_0} \Big(\ln C_{|\rho(i)|}^{\widetilde{a}_i(\vz)}+\widetilde{a}_i(\vz)\ln j(\vz)\Big)
+ \sum_{i=l_0+1}^{\infty}\Big(\ln C_{|\rho(i)|}^{\widetilde{a}_i(\vz)}+\widetilde{a}_i(\vz)\ln j(\vz)\Big)\\
&\lesssim \sum_{i=1}^{l_0}2^i\ln j(\vz)+\sum_{i=l_0+1}^{\infty} \frac{\ln\vz^{-1}}{2^{(i-1)\frac{1-s}{s}}h(2^{l_0})}
\Big(i-l_0+\ln \frac{e\,i}{l_0}+\ln j(\vz)\Big)\\
&\lesssim 2^{l_0}\ln j(\vz)+\frac{\ln\vz^{-2}}{2^{(l_0-1)\frac{1-s}{s}}h(2^{l_0})}\big(1+\ln
j(\vz)\big)\lesssim 2^{l_0}\ln j(\vz).
\end{align*}

Now we want to  show that
$2^{\frac{l_0}{1-s}}\lesssim\widetilde{d}(\vz)\le d(\vz),$ where
$$\widetilde{d}(\vz):=\max\{j: \widetilde{\gamma}_j>\vz^2\}\le \max\{j: {\gamma}_j>\vz^2\}=:d(\vz).$$
Note that
\begin{align*}
\widetilde{d}(\vz)
&=\max\{j: \ln \frac{1}{\widetilde{\gamma}_j\lambda_2}<\ln\vz^{-2}-\ln\frac{1}{\lambda_2}\}\\
&=\max\{j: j^{\frac{1-s}{s}}h(j)<\ln\vz^{-2}-\ln\frac{1}{\lambda_2}\}.
\end{align*}
Choose a sufficiently large positive number $c_0>1$, such that for
sufficient small $\vz>0$,
$$\ln \Big(\widetilde{\gamma}_{2^{\frac{l_0}{1-s}-c_0}}\lambda_2\Big)^{-1}=\big(2^{\frac{l_0}{1-s}-c_0}\big)^{\frac{1-s}{s}}h\big(2^{\frac{l_0}{1-s}-c_0}\big).$$
From \eqref{3.3} and \eqref{3.5}, we get
\begin{align*}
\ln \Big(\widetilde{\gamma}_{2^{\frac{l_0}{1-s}-c_0}}\lambda_2\Big)^{-1} &\le 2^{\frac{l_0}{s}}\cdot 2^{-c_0\cdot\frac{1-s}{s}}\cdot h(2^{l_0})\cdot \frac{\ln\big(1+2^{\frac{l_0}{1-s}}-c_0\big)}{\ln(2^{l_0}+1)}\\
&\le 2^{\frac{l_0}{s}}h(2^{l_0})\cdot 2^{-c_0\cdot\frac{1-s}{s}}\cdot \frac{1}{1-s}\\
&\le \frac{1}{2}\ln\vz^{-2}<\ln\vz^{-2}-\ln\frac{1}{\lambda_2}.
\end{align*}
It implies that
$$2^{\frac{l_0}{1-s}-c_0}\le \widetilde{d}(\vz)\le d(\vz),\ \ {\rm
or,\ equivalently}\ \ \ 2^{l_0}\lesssim d(\vz)^{1-s}.$$
Then by \eqref{1.11} we obtain
\begin{align*}
0&\le \lim_{\vz^{-1}+d\to\infty}\frac{\ln
n(\vz,S_{d,\bsgamma})}{(\ln\vz^{-1})^s+d^t}
\lesssim\lim_{\vz^{-1}\to \infty}\frac{2^{l_0}\ln j(\vz)}{(\ln \vz^{-1})^s}\\
&\lesssim \lim_{\vz^{-1}\to \infty}\frac{d(\vz)^{1-s}\ln
j(\vz)}{(\ln \vz^{-1})^s}=0.
\end{align*}
This means that in the case of $\ln\vz^{-1}\ge d^t$,
EXP-$(s,t)$-WT holds.

Theorem \ref{thm1}  is proved. $\hfill\Box$

\

\noindent{\it Proof of Corollary \ref{cor}.}

(1) Assume that EXP-$(s,t)$-WT holds with $0<s<1$ and $0<t\le 1$.
 By Theorem \ref{thm1}, we get \eqref{1.10} and \eqref{1.11}. According to Lemma \ref{lem2},
 \eqref{1.10} is equivalent to \eqref{1.12}. Hence,
 \eqref{1.12} holds.
 By  \eqref{1.11} and the fact $j(\vz)\ge 2$ for $\vz^2<\lz_2$,  we get \eqref{2.4}.
 By Lemma \ref{lem3}, we obtain \eqref{1.13}.

(2) Assume that \eqref{1.12} and \eqref{1.14} hold. By Lemma
\ref{lem2} and Theorem \ref{thm1} it suffices to prove
\eqref{1.11}. According to Lemma \ref{lem3} and \eqref{1.14}, we
obtain
$$\lim_{\vz\to0}\frac{d(\vz)^{1-s+s\delta }}{(\ln\vz^{-1})^s}=0. $$
By Lemma 2 and \eqref{1.12}, it follows that
$$\lim_{\vz\to0}\frac{d(\vz)^{1-s}\ln j(\vz)}{(\ln\vz^{-1})^s}=
\lim_{\vz\to0}\frac{d(\vz)^{1-s+s\delta }}{(\ln\vz^{-1})^s}\
\lim_{\vz\to0}\frac{\ln j(\vz)}{(\ln \vz^{-1})^{s\delta}} =0.
$$
Hence, \eqref{1.11} holds. Corollary \ref{cor} is proved.
 $\hfill\Box$

\

\noindent{\it Proof of Theorem \ref{thm2}.}

Assume  that EXP-UWT holds. Then for any $0<s<1$ and $0<t\le 1$,
EXP-$(s,t)$-WT also holds. According to Corollary $\ref{cor}$,
$$\lim_{j\rightarrow\infty}\frac{\ln(\ln\frac{1}{\lambda_j})}{\ln(\ln j)}=\infty\ \ \ \mbox{and}\ \ \
\lim_{j\rightarrow\infty}\frac{\ln\frac{1}{\gamma_j}}{j^\frac{1-s}{s}}=\infty\
\ {\rm for\ all}\ s\in(0,1).$$ Using Lemma  $\ref{lem5}$, we
obtain $\eqref{1.12}$ and \eqref{1.15}.

\indent On the other hand, assume that  $\eqref{1.12}$ and
\eqref{1.15} hold. We want to show EXP-UWT holds,  i.e.,
EXP-$(s,t)$-WT holds  for any $s, t\in (0,1)$. It follows from
\eqref{1.15} and Lemma \ref{lem5} that \eqref{1.14} holds for any
$s\in (0,1)$. By Corollary $\ref{cor}$ we know that EXP-$(s,t)$-WT
holds for any $s,t\in (0,1)$. Hence, EXP-UWT holds.

We can also use Lemma \ref{lem1} to prove that EXP-UWT holds.  Indeed,
 it suffices to show that EXP-$(s,t)$-WT holds for any $s,t\in (0,1)$.
 Let $V:=\max\{d^t,(\ln\vz^{-1})^s\}$, then $V\to\infty$ as $d+\varepsilon^{-1} \rightarrow
 \infty$. Using the same method as in Lemma \ref{lem2} we obtain  that
$$\ln j(\vz)\le  \max\{\ln J_1,\ (2\ln\vz^{-1})^{\frac{s}{3}}\}\le \max\{\ln J_1, (2V)^{1/3}\}, $$
and $$d(\vz)\le \max\{J_2, (2\ln \vz^{-1}) ^{\frac{s}{3}}\}\le
\max\{ J_2, (2V)^{1/3}\},$$where $J_1,J_2$ are  the positive
constants. It follows from Lemma \ref{lem1} that
\begin{align*} &\quad\ \lim_{d+\varepsilon^{-1} \rightarrow \infty}
\frac{\ln n(\vz, S_{d,\bsgamma})}{d^t +(
\ln\varepsilon^{-1})^s}\le \lim_{V \rightarrow \infty}
\frac{d(\varepsilon)\ln j(\varepsilon)}{V}\\
& \le \lim_{V\rightarrow \infty}\frac{\max\{\ln J_1,
(2V)^{1/3}\}\cdot \max\{J_2, (2\ln \vz^{-1}) ^{\frac{s}{3}}\}
}{V}=0.
\end{align*}

 Theorem \ref{thm2} is proved. $\hfill\Box$

\

\noindent\textbf{Acknowledgments}
  The authors  were
supported by the National Natural Science Foundation of China
(Project no. 12371098).

\end{document}